\documentclass[onefignum,onetabnum,final]{siamart190516}

\ifpdf
\hypersetup{
  pdftitle={On the Microscopic Modeling of Vehicular Traffic on General Networks},
  pdfauthor={R.M.~Colombo, H.~Holden and F.~Marcellini}
}
\fi

\usepackage{amsfonts} 
\usepackage{enumerate}
\usepackage{upref}
\usepackage{graphicx}

\newsiamremark{remark}{\textsc{Remark}}

\headers{Microscopic Modeling of Traffic on Networks}{R.M.~Colombo, H.~Holden and F.~Marcellini}

\newenvironment{proofof}[1]{\smallskip\noindent{\textbf{Proof~of~#1.}}%
  \hspace{1pt}}{\hspace{-5pt}{\nobreak\quad\nobreak\hfill\nobreak%
    $\square$\vspace{2pt}\par}\smallskip\goodbreak}

\renewcommand{\L}[1]{{\mathbf{L}^#1}}

\newcommand{\W}[2]{{\mathbf{W}^{#1,#2}}}

\newcommand{\modulo}[1]{{\left|#1\right|}}

\newcommand{\reali}{{\mathbb{R}}}

\renewcommand{\epsilon}{\varepsilon}
\renewcommand{\phi}{\varphi}
\renewcommand{\theta}{\vartheta}

\title{On the Microscopic Modeling of \\
  Vehicular Traffic on General Networks}

\author{Rinaldo M.~Colombo\thanks{INdAM Unit, University of Brescia (\email{rinaldo.colombo@unibs.it}, \url{http://rinaldo.unibs.it}).} \and Helge Holden\thanks{Department of Mathematical Sciences, NTNU Norwegian
  University of Science and Technology (\email{helge.holden@ntnu.no}, \url{https://www.ntnu.edu/employees/holden}).} \and Francesca
  Marcellini\thanks{INdAM Unit, University of Brescia (\email{francesca.marcellini@unibs.it}).}}

\allowdisplaybreaks

\begin{document}

\maketitle

\begin{abstract}
  We introduce a formalism to deal with the microscopic modeling of
  vehicular traffic on a road network. Traffic on each road is
  uni-directional, and the dynamics of each vehicle is described by a
  Follow-the-Leader model. From a mathematical point of view, this
  amounts to define a system of ordinary differential equations on an
  arbitrary network. A general existence and uniqueness result is
  provided, while priorities at junctions are shown to hinder the
  stability of solutions. We investigate the occurrence of the Braess
  paradox in a time-dependent setting within this model.  The
  emergence of Nash equilibria in a non-stationary situation results
  in the appearance of Braess type paradoxes, and this is supported by
  numerical simulations.
\end{abstract}



\section{Introduction}
\label{sec:I}

The literature on the modeling of vehicular traffic has been growing
very quickly in recent years. A variety of approaches coexists,
typically they can be characterized as either macroscopic or
microscopic.

The former ones are usually based on partial differential equations,
their prototype being the Lighthill--Whitham~\cite{LighthillWhitham}
and Richards~\cite{Richards} model. Deep criticisms~\cite{Daganzo1}
led to the formulation of entirely new continuum models, such
as~\cite{AwRascle}, or multiphase
models~\cite{BlandinWorkGoatinPiccoliBayen, Colombo1.5,
  ColomboMarcelliniRascle, Goatin, Marcellini2017} and models on
networks, starting from~\cite{HoldenRisebroTraffic} up to the recent
monograph~\cite{GaravelloHanPiccoliBook}.

Microscopic models also have a long tradition,
see~\cite{GazisHermanRothery}. They are usually denoted as
Follow-the-Leader models, the dynamics being governed by the
interaction between a vehicle and the vehicle immediately in front of
it.  More precisely, we have
\begin{equation*}
  \dot x_\alpha = v\left(\frac{\ell}{x_{\alpha-1} - x_\alpha}\right),
\end{equation*}
where $x_\alpha<x_{\alpha-1}$ denotes the position of two consecutive
vehicles, each of length $\ell$, driving with a velocity function $v$.

Various connections between the two scalings are found in the
literature, referring to limiting procedures yielding the macroscopic
models as limit of the microscopic ones, as
in~\cite{AwKlarMaterneRascle, DiFrancescoRosini,
  HoldenRisebro2018_uno, HoldenRisebro2018_due}, or mixing the two
scales~\cite{ColomboMarcelliniMixed, ColomboMarcelliniAware,
  LattanzioPiccoli}. Note however that most macroscopic models
prescribe traffic rules at junctions that also require some sort of
flow maximization, see~\cite{GaravelloHanPiccoliBook} for more
details. In the construction below, no such maximization is used, and
this will make a continuum limit more complicated.  However, the
chosen priority rules are sufficient to single out a unique evolution.
Other approaches have been studied in the literature.

Apart from models based on differential equations, many other
mathematical tools are used in the literature to describe traffic on
networks and, where possible, to account for Braess paradox. For
instance, a stochastic approach can be found in~\cite{MR3812279}, an
evolutionary variational inequality model is studied
in~\cite{MR2351220}, while queue theory is applied
in~\cite{LinHong2009}. The assessment of the network performance due
to selfish routing can be found in~\cite{Roughgarden2005}. In contrast
to these approaches, here the dynamics is fully described by ODEs,
with simple priority rules at junctions.

Modern vehicular traffic offers a plethora of modeling challenges --
complicated network geometries, roundabouts, traffic lights, traffic
obstructions, a combination of various agents (pedestrians,
bicyclists, a wide range of different vehicles), noise, pollution,
etc. We here focus on a general network with only one type of
vehicles, but we provide a consistent and rigorous model for behavior
at junctions based on a Follow-the-Leader model. See
also~\cite{MR3541527, MR3905826} for related work.

As far as we are aware of, the microscopic modeling of traffic on a
network has not been formalized systematically before.

Our approach yields a model that comprises a system of (discontinuous)
ordinary differential equations (ODEs) on a network with a concrete
behavior at junctions. Moreover, the present model comprises the
presence of different priorities between roads. Below, we present a
framework where rigorous statements about the microscopic modeling of
vehicle dynamics, complying with priority rules, can be formalized,
proved, and numerically computed.

Within this structure, we formalize an ODE-based model and provide an
existence and uniqueness result for the corresponding evolution, see
Theorem~\ref{thm:1}. By means of an example, we show that the usual
well-posedness estimates may not hold. Indeed, and consistently with
everyday experience, small changes in the departure time of a single
vehicle may lead to large changes in the arrival time of that vehicle,
due for instance to arriving slightly earlier or later at junctions
where priority has to be yielded, see Remark~\ref{rem:no}.

\medskip

A main aim for us has been to investigate the ubiquitous Braess
paradox in a time-dependent setting through deterministic differential
equations.  As far as we know, in this context, the Braess phenomenon
has so far only been analyzed mathematically in the stationary case.
\begin{figure}[h!]
  \centering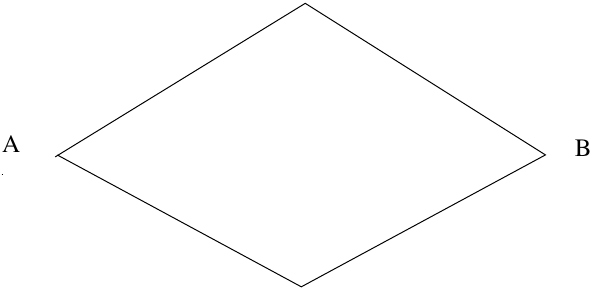\qquad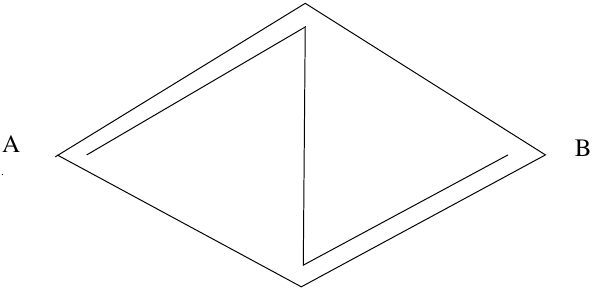
  \caption{Left: network consisting of two routes connecting $A$ to
    $B$. The first route consists of the roads $3$ and $6$; the second
    route consists of the roads $2$ and $5$. Right: network consisting
    of three routes connecting $A$ to $B$. The first route consists of
    the roads $3$ and $6$; the second route consists of the roads $2$
    and $5$. The third route consists of the roads $3$, $4$, and $5$.}
  \label{fig:cl}
\end{figure}
Recall first the simplest example of Braess paradox.  We have a
network consisting of two routes connecting $A$ to $B$, where the
first route consists of the roads $3$ and $6$, while the second route
consists of the roads $2$ and $5$, see Figure~\ref{fig:cl}
(left). Traffic is unidirectional in the direction from $A$ to
$B$. The roads $2$ and $6$ are equal, with unlimited capacity, and the
travel time is $45$ minutes independently of the number of
vehicles. The roads $3$ and $5$ are also equal and the travel time is
$N/100$, where $N$ is the number of vehicle traveling on the road. We
suppose that $4000$ vehicles move from $A$ to $B$. Each driver chooses
the fastest route and the resulting Nash equilibrium amounts to $2000$
drivers traveling along each road. Correspondingly, we find a travel
time of $65$ minutes for each driver.

Then, we add a new road, say number $4$, as in Figure~\ref{fig:cl}
(right), characterized by a negligible travel time. Drivers start
using the new road choosing the route consisting of roads $3$, $4$,
and $5$, reducing their travel time. However, since the new route
$[3,4,5]$ is more convenient than both $[3,6]$ and $[2,5]$, more and
more drivers choose this new route. As a result, the travel time
increases to $80$ minutes for everyone. This is the paradox: contrary
to common sense, adding a new road to a network may make travel times
worse for everyone.

This paradox was introduced by Braess in 1968~\cite{BraessParadox}
with a different example, see also~\cite{Nagurney}, and it has been
observed in real situations. In 1968, for instance, a highway segment
was closed in Stuttgart and traffic improved, see~\cite{K}. In 1990,
in New York the $42$nd street was closed for one day and, again
unexpectedly, traffic improved, see~\cite{NY}.

This paradox appears in other situations as well, not only modeling
vehicular traffic.  In crowd dynamics, the well-known phenomenon of
reducing the evacuation time from a closed space by suitably
positioning obstacles near exits that direct the crowd movement (and
closing a number of paths) is described through a partial differential
equation model in~\cite{CGM}.

Our aim is to capture the Braess paradox in a non-stationary setting
in the present Follow-the-Leader model.  For simplicity we study the
case of the network depicted in Figure~\ref{fig:cl}.  The present
framework allows us to show the dynamic emergence of a Braess-like
situation in a fully non-stationary setting. In contrast to the
examples typically found in the literature~\cite{BraessParadox,
  ColomboHolden2016, Nagurney}, in the examples below we start from an
\emph{empty} network. As vehicles enter it, the measured travel times
show the rise of Braess paradox, as shown by numerical computations.

A key role is here played by our postulating the behavior of drivers
as described by a Nash equilibrium. Indeed, we view drivers as players
competing in a non-cooperative way to reduce their travel times, see
also~\cite{BressanHan2011, BressanHan2012, ColomboHolden2016}. In
particular situations, the solution of the Follow-the-Leader model at
Nash equilibrium leads to the emergence of non-stationary Braess-like
situations as supported numerically.

\smallskip

The next section is devoted to the definition of the microscopic model
on a network. Section~\ref{sec:BP} is devoted to the emergence of
Braess paradox, obtained as Nash equilibrium within the framework of
the model here introduced. The last section collects the analytic
proofs.

\section{Formal Framework}
\label{sec:AF}

The standard first-order \emph{Follow-the-Leader} model is based on
the following Cauchy problem for a system of ordinary differential
equations:
\begin{equation}
  \label{eq:1}
  \left\{
    \begin{array}{l@{\qquad}r@{\;}c@{\;}l}
      \dot x_1 = V_{\max}
      \\
      \dot x_\alpha = v\left(\frac{\ell}{x_{\alpha-1} - x_\alpha}\right)
      & \alpha
      & \in
      & \{2, \ldots, n\} \,,
      \\
      x_\alpha(0) = x^o_\alpha
      & \alpha
      & \in
      & \{1, \ldots, n\} \,.
    \end{array}
  \right.
\end{equation}
Here, $n$ drivers labeled by their positions $x_1, \ldots x_n$ drive
at speed $v \left(\ell/ (x_{\alpha-1} - x_\alpha)\right)$, where
$\ell$ is the length of each vehicle and the speed $v$ satisfies the
condition:
\begin{description}
\item[(SpeedLaw)] $v$ is a Lipschitz continuous function and attains
  values in $[0, V_{\max}]$, i.e.,
  $v \in \W{1}{\infty} (\reali^+; [0, V_{\max}])$, and it is a
  (weakly) decreasing function such that $v (\rho) = 0$ for all
  $\rho \geq 1$.
\end{description}
The constant $V_{\max}$ is an upper bound for the speed of all
vehicles.  The drivers' initial positions are $x^o_0, \ldots,
x^o_n$. It is well-known that the assumption
$x^o_\alpha - x^o_{\alpha-1} \geq \ell$ for
$\alpha \in \{1, \ldots, n\}$ ensures that the solutions
to~\eqref{eq:1} keep satisfying the same bound, i.e.,
$x_\alpha (t) - x_{\alpha-1} (t) \geq \ell$ for all $\alpha$ and for
all $t \geq 0$, meaning that no collision ever occurs.

\bigskip

We now introduce a formalism to deal with the extension
of~\eqref{eq:1} to a general network.

\subsubsection*{Network Structure} The network is a collection of $m$
real intervals: each of them representing a road. Roads are of three
types:
\begin{description}

\item[$\bullet$~Entry Roads:] they are copies of the (open) half--line
  $\left]-\infty, 0\right[$;

\item[$\bullet$~Middle Roads:] they are bounded intervals of the type
  $[0, L_j[$, where $L_j > 0$ is the road length;

\item[$\bullet$~Exit Roads:] they are copies of the half--line
  $\left[0, +\infty\right[\,$.
\end{description}
Entry Roads and Exit Roads have infinite length. We assume throughout
that the vehicle length $\ell$ is negligible with respect to the
(finite) length of each Middle Road: $\ell \ll L_j$ for all $j$
indexing a Middle Road.

To simplify various expressions, it is convenient to assign $L_j = 0$
for all $j$ indexing an Entry Road. It can also be of use to set
$L_j = +\infty$ for each Exit Road. This convention allows us to
introduce the following terminology, of use below: for each Middle
Road or Entry Road $j$, the \emph{end of the road} is the real
interval $\left]L_j - \ell, L_j \right[$. Here, to define the end of
the road we use the vehicle length $\ell$ but choosing a different
length $\ell'$, with $\ell' > \ell$, is also possible.

Road indices are assigned so that whenever two or more roads enter the
same junction, drivers on roads with lower indices have priority.
\begin{figure}[h!]
  \centering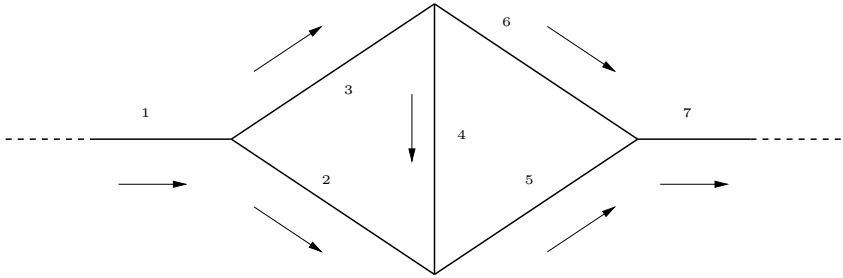
  \caption{\label{fig:Braess} The network notation in the case of
    Braess network, see~\cite{BraessParadox}. Note that roads are
    numbered so that at each junction, roads coming from the right
    have the priority.}
\end{figure}
Throughout, we assume that junctions either have a single incoming
road, or have a single outgoing road. The case of general junctions
with several incoming and outgoing roads can be treated by the same
methods described below, at the cost of a more intricate formalism.

\subsubsection*{Drivers' Route Choices} The $n$ drivers are indexed by
$\alpha$, running between $1$ and $n$. Each driver's route is
identified by the sequence of the indices of the roads that constitute
the route. We denote by $\mathcal{R}_\alpha$ the route followed by
driver $\alpha$. For instance, with reference to the Braess network in
Figure~\ref{fig:Braess}, the route followed by the driver $\alpha=1$
choosing the \emph{``lower''} route is identified by
$\mathcal{R}_1 = [1,2,5,7]$. If the driver $\alpha=2$ follows the
route passing through the road $4$, then
$\mathcal{R}_2 = [1,3,4,5,7]$.

Throughout, $r_\alpha (t)$ stands for the index of the road along
which the $\alpha$th driver is traveling at time $t$. We also write
$j' = \mathcal{N}_\alpha (j)$ meaning that the $\alpha$th driver at
the end of the $j$th road enters the $j'$th road. For instance, with
reference to Figure~\ref{fig:Braess}, if the route of the driver
$\alpha=1$ is $\mathcal{R}_1 = [1,3,6,7]$, then we have
$\mathcal{N}_1 (1) = 3$, $\mathcal{N}_1 (3) = 6$, and
$\mathcal{N}_1 (6) = 7$.

Along each road, we identify the $\alpha$th driver's position through
the time dependent variable $x_\alpha$ ranging in
$\left]-\infty, 0\right[$ along Entry Roads, in $\left[0, L_j\right[$
along Middle Roads and in $\left[0, +\infty\right[$ along Exit Roads.

A key assumption in the construction below amounts to require that no
\emph{loop} is possible for any driver:
\begin{description}
\item[(NoLoop):] No route can contain the same road twice.
\end{description}
Note that the network itself may well contain loops, but
condition~\textbf{(NoLoop)} requires that none of them can be part of
a route.

Of use below is also the following, quite natural, requirement:
\begin{description}
\item[(NoDeadEnd)] The last road in each route is an Exit Road.
\end{description}

\subsubsection*{Drivers' Speed} We now specify the speed chosen by the
$\alpha$th driver, depending on the position and on that of the
vehicles preceding the driver.  We consider several special cases.

\subsubsection*{Far from Junctions} At time $t$ the driver is
positioned at $x_\alpha (t)$ driving along road $j = r_\alpha (t)$. As
long as the $\alpha$th driver is not at the end of the road indexed by
$r_\alpha (t)$, i.e.,~$x_\alpha (t) < L_{r_\alpha (t)} - \ell$, the
speed only depends on the free space ahead, similarly to what happens
in~\eqref{eq:1}:
\begin{equation}
  \label{eq:3}
  \begin{array}{@{}r@{}c@{}l@{}}
    \dot x_\alpha
    & =
    & \left\{
      \begin{array}{@{}l@{\,}l@{}}
        V_{r_\alpha(t)}
        & \mbox{if}
          \left\{
          \begin{array}{@{}l@{}}
            \alpha \mbox{ is not at the end of the road;}
            \\
            \mbox{no one is on the same road in front of } \alpha .
          \end{array}
        \right.
        \\[10pt]
        v_{r_\alpha (t)} \! \! \left(\frac{\ell}{p-x_\alpha}\right)
        & \mbox{if}
          \left\{
          \begin{array}{@{}l@{}}
            \mbox{someone is on the same road in front of } \alpha;
            \\
            p
            \mbox{ is the position of the nearest vehicle in front of}
            \\
            \alpha \mbox{  on the same road.}
          \end{array}
        \right.
      \end{array}
    \right.
    \\[30pt]
    & =
    & \left\{
      \begin{array}{@{}l@{\,}l@{}}
        V_{r_\alpha(t)}
        & \mbox{if}
          \left\{
          \begin{array}{@{}l@{}}
            x_\alpha (t) < L_{r_\alpha (t)}  - \ell;
            \\
            \left\{
            \alpha' \in \{1, \ldots, n\} \colon
            r_{\alpha'} (t) {=} r_\alpha (t) \mbox{ and }
            x_{\alpha'} (t) {>} x_\alpha (t)
            \right\} = \emptyset.
          \end{array}
        \right.
        \\[10pt]
        v_{r_\alpha (t)} \! \! \left(\frac{\ell}{p-x_\alpha}\right)
        & \mbox{if}
          \left\{
          \begin{array}{@{}l@{}}
            \left\{
            \alpha' \in \{1, \ldots, n\} \colon
            r_{\alpha'} (t) {=} r_\alpha (t) \mbox{ and }
            x_{\alpha'} (t) {>} x_\alpha (t)
            \right\} \neq \emptyset;
            \\
            p {=} \min\left\{
            x_{\alpha'} {\in} [0, L_{r_\alpha (t)}[ {\colon}
            r_{\alpha'} (t) {=} r_\alpha (t) \mbox{ and }
            x_{\alpha'} (t) {>} x_\alpha (t)
            \right\} \!.
          \end{array}
        \right.
      \end{array}
    \right.
  \end{array}
  \!\!\!\!\!
\end{equation}
Indeed, the set
$\left\{ \alpha' \in \{1, \ldots, n\} \colon r_{\alpha'} (t) =
  r_\alpha (t) \mbox{ and } x_{\alpha'} (t) > x_\alpha (t) \right\}$
identifies the (indices $\alpha'$ of) drivers preceding $\alpha$ along
the road $r_\alpha (t)$ where $\alpha$ is driving at time $t$. If no
such driver exists, $\alpha$ drives at the maximal speed
$V_{r_\alpha(t)}$ possible along the road $r_\alpha (t)$. On the other
hand, if
$\{ \alpha' \in \{1, \ldots, n\} \colon r_{\alpha'} (t) = r_\alpha (t)
\mbox{ and } x_{\alpha'} (t) > x_\alpha (t) \} \neq \emptyset$, then
the speed $\dot x_\alpha (t)$ of the $\alpha$th driver is adjusted to
the distance between $\alpha$ and the driver at position $p$, who is
the one immediately in front of $\alpha$, as usual in a
\emph{Follow-the-Leader} model.

Note that if $r_\alpha (t)$ is an Exit Road, then we understand that
the condition $x_\alpha < L_{r_\alpha (t)}- \ell$ is true for all
$x_\alpha$.

\subsubsection*{A Fork in the Road} Consider a junction with one road
(either an Entry or a Middle Road) entering it and any number of roads
exiting it. At time $t$ driver $\alpha$ is close to the end of the
Entry Road or the Middle Road $r_\alpha (t)$, in the sense that
$x_\alpha (t) \in [ L_{r_\alpha (t)} - \ell, L_{r_\alpha
  (t)}[$. Driver $\alpha$ chooses the speed $\dot x_\alpha (t)$ taking
into consideration only those drivers preceding him/her along the road
$r_\alpha (t)$ or present in the next road
$\mathcal{N}_\alpha (r_\alpha (t))$ he/she is going to take, see
Figure~\ref{fig:1to2}.

We then set
\begin{displaymath}
  \begin{array}{@{}r@{}c@{}l@{}}
    \dot x_\alpha
    & =
    & \left\{
      \begin{array}{@{}l@{\,}l@{}}
        V_{r_\alpha(t)}
        & \mbox{if}
          \left\{
          \begin{array}{@{\,}l@{}}
            \alpha\mbox{ is at the end of the road;}
            \\
            \mbox{no one is on the same road in front of }\alpha;
            \\
            \mbox{no one is on the road where $\alpha$ is going.}
          \end{array}
        \right.
        \\[20pt]
        v_{r_\alpha (t)} \! \! \left(\frac{\ell}{p + L_{r_\alpha (t)}-x_\alpha}\right)
        & \mbox{if}
          \left\{
          \begin{array}{@{}l@{}}
            \alpha\mbox{ is at the end of the road;}
            \\
            \mbox{no one is on the same road in front of }\alpha;
            \\
            \mbox{someone is on the road where $\alpha$ is going;}
            \\
            p
            \mbox{ is the position of the nearest vehicle in front of}
            \\
            \alpha \mbox{  on the same road.}
          \end{array}
        \right.
      \end{array}
    \right.
  \end{array}
\end{displaymath}
\begin{equation}
  \label{eq:4}
  \begin{array}{@{}r@{}c@{}l@{}}
    \hphantom{\dot x_\alpha}
    & =
    & \left\{
      \begin{array}{@{}l@{\,}l@{}}
        V_{r_\alpha(t)}
        & \mbox{if}
          \left\{
          \begin{array}{@{\,}l@{}}
            x_\alpha (t) > L_{r_\alpha (t)}  - \ell;
            \\
            \left\{
            \alpha' {\in} \{1, \ldots, n\} \colon
            r_{\alpha'} (t) {=} r_\alpha (t) \mbox{ and }
            x_{\alpha'} (t) {>} x_\alpha (t)
            \right\} {=} \emptyset;
            \\
            \left\{
            \alpha' {\in} \{1, \ldots, n\} \colon
            r_{\alpha'} (t) {=} \mathcal{N}_\alpha\left(r_\alpha (t)\right)
            \right\} {=} \emptyset.
          \end{array}
        \right.
        \\[20pt]
        v_{r_\alpha (t)} \! \! \left(\frac{\ell}{p + L_{r_\alpha (t)}-x_\alpha}\right)
        & \mbox{if}
          \left\{
          \begin{array}{@{}l@{}}
            x_\alpha (t) > L_{r_\alpha (t)}  - \ell;
            \\
            \left\{
            \alpha' {\in} \{1, \ldots, n\} \colon
            r_{\alpha'} (t) {=} r_\alpha (t) \mbox{ and }
            x_{\alpha'} (t) {>} x_\alpha (t)
            \right\} {=} \emptyset;
            \\
            \left\{
            \alpha' {\in} \{1, \ldots, n\} \colon
            r_{\alpha'} (t) {=} \mathcal{N}_\alpha\left(r_\alpha (t)\right)
            \right\} \neq \emptyset;
            \\
            p {=} \min\left\{
            x_{\alpha'} {\in} [0, L_{\mathcal{N}_\alpha\left(r_\alpha (t)\right)}] \colon
            r_{\alpha'} (t) {=} \mathcal{N}_\alpha\left(r_\alpha (t)\right)
            \right\}.
          \end{array}
        \right.
      \end{array}
    \right.
  \end{array}
\end{equation}
Indeed, when
$\{ \alpha' \in \{1, \ldots, n\} \colon r_{\alpha'} (t) =
\mathcal{N}_\alpha\left(r_\alpha (t)\right) \}$ is empty, no one is
preceding the $\alpha$th driver along his/her route and the $\alpha$th
driver proceeds at full speed. On the other hand, if
$\{ \alpha' \in \{1, \ldots, n\} \colon r_{\alpha'} (t) =
\mathcal{N}_\alpha\left(r_\alpha (t)\right) \} \neq \emptyset$, then
the driver immediately preceding $\alpha$ is at position $p$, as
defined in~\eqref{eq:4}.
\begin{figure}[h!]
  \centering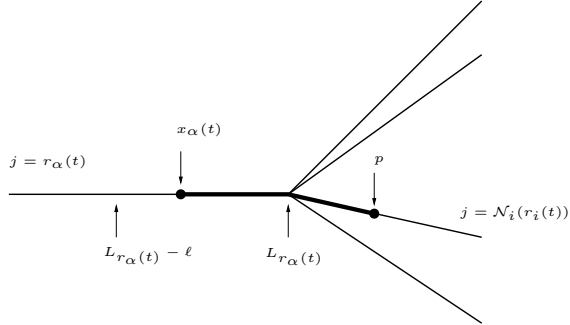
  \caption{\label{fig:1to2}Notation used in~\eqref{eq:4}. The
    $\alpha$th driver, located at $x_\alpha (t)$ is approaching the
    end of the $r_\alpha (t)$ road, in the sense that
    $x_\alpha (t) \in [L_{r_\alpha (t)}-\ell,L_{r_\alpha (t)}]$, and
    its predecessor is at $p$ along the road
    $\mathcal{N}_\alpha\left(r_\alpha (t)\right)$.}
\end{figure}
The resulting speed $\dot x_\alpha (t)$ of the $\alpha$th driver is
then chosen according to the usual \emph{Follow-the-Leader} rule, with
$p +L_{r_\alpha (t)} - x_\alpha (t)$ being the physical distance
measured along the road between the $\alpha$th driver and his/her
predecessor, see Figure~\ref{fig:1to2}.

\subsubsection*{Roads Merging} Consider now a junction with several
roads entering a single road. We assume that the roads' indexing
respects the roads' priorities, in the sense that if the roads $j$ and
$j'$ enter the same junction and $j < j'$, then the drivers on the
road $j$ have priority over those on road $j'$. Call $J$ the set of
indices of the roads entering the junction under consideration.

First, we deal with the case of a driver coming from the road that has
the priority over all the other incoming roads. In this case, we have
$r_\alpha (t) = \min J$ by assumption. We then set
\begin{displaymath}
  \begin{array}{@{}r@{}c@{}l@{}}
    \dot x_\alpha
    & =
    & \left\{
      \begin{array}{@{}l@{}l}
        V_{r_\alpha(t)}
        & \mbox{if}
          \left\{
          \begin{array}{@{}l}
            \alpha\mbox{ is at the end of the road;}
            \\
            \alpha\mbox{'s road has the priority;}
            \\
            \mbox{no one is on the same road in front of }\alpha;
            \\
            \mbox{no one is on the road where $\alpha$ is going.}
          \end{array}
        \right.
        \\[25pt]
        v_{r_\alpha (t)} \! \! \left(\frac{\ell}{p + L_{r_\alpha (t)}-x_\alpha}\right)
        & \mbox{if}
          \left\{
          \begin{array}{@{}l}
            \alpha\mbox{ is at the end of the road;}
            \\
            \alpha\mbox{'s road has the priority;}
            \\
            \mbox{no one is on the same road in front of }\alpha;
            \\
            \mbox{someone is on the road where $\alpha$ is going;}
            \\
            p
            \mbox{ is the position of the nearest vehicle in front of}
            \\
            \alpha \mbox{  on the same road.}
          \end{array}
        \right.
      \end{array}
    \right.
  \end{array}
\end{displaymath}
\begin{equation}
  \label{eq:2}
  \begin{array}{@{}r@{}c@{}l@{}}
    \hphantom{\dot x_\alpha}
    & =
    & \left\{
      \begin{array}{@{}l@{}l}
        V_{r_\alpha(t)}
        & \mbox{if}
          \left\{
          \begin{array}{@{}l}
            x_\alpha (t) {>} L_{r_\alpha (t)}  - \ell;
            \\
            r_\alpha (t) {=} \min J;
            \\
            \left\{
            \alpha' {\in} \{1, \ldots, n\} \colon
            r_{\alpha'} (t) {=} r_\alpha (t) \mbox{ and }
            x_{\alpha'} (t) {>} x_\alpha (t)
            \right\} {=} \emptyset;
            \\
            \left\{
            \alpha' {\in} \{1, \ldots, n\} \colon
            r_{\alpha'} (t) {=} \mathcal{N}_\alpha\left(r_\alpha (t)\right)
            \right\} {=} \emptyset.
          \end{array}
        \right.
        \\[25pt]
        v_{r_\alpha (t)} \! \!
        \left(\frac{\ell}{p + L_{r_\alpha (t)}-x_\alpha}\right)
        & \mbox{if}
          \left\{
          \begin{array}{@{}l}
            x_\alpha (t) {>} L_{r_\alpha (t)}  - \ell;
            \\
            r_\alpha (t) {=} \min J;
            \\
            \left\{
            \alpha' {\in} \{1, \ldots, n\} \colon
            r_{\alpha'} (t) {=} r_\alpha (t)
            \mbox{ and } x_{\alpha'} (t) {>} x_\alpha (t)
            \right\} {=} \emptyset;
            \\
            \left\{
            \alpha' {\in} \{1, \ldots, n\} \colon
            r_{\alpha'} (t) {=} \mathcal{N}_\alpha\left(r_\alpha (t)\right)
            \right\} {\neq} \emptyset;
            \\
            p = \min\left\{
            x_{\alpha'} {\in} [0, L_{r_\alpha (t)}] \colon
            r_{\alpha'} (t) {=} \mathcal{N}_\alpha\left(r_\alpha (t)\right)
            \right\}.
          \end{array}
        \right.
      \end{array}
    \right.
  \end{array}
\end{equation}
Similarly to the previous case of the fork in the road, i.e.,
equation~\eqref{eq:4},
$\{ \alpha' \in \{1, \ldots, n\} \colon r_{\alpha'} (t) =
\mathcal{N}_\alpha\left(r_\alpha (t)\right) \}$ is empty whenever the
$\alpha$th driver has free road ahead. When
$\{ \alpha' \in \{1, \ldots, n\} \colon r_{\alpha'} (t) =
\mathcal{N}_\alpha\left(r_\alpha (t)\right) \}$ is nonempty, $p$ as
defined in~\eqref{eq:2} is the position of the first driver in front
of $\alpha$, and $p + L_{r_\alpha (t)} - x_\alpha (t)$ is the length
of the free road in front of the driver $\alpha$, see
Figure~\ref{fig:3t01} (right).

\medskip

Let now the $\alpha$th driver approach the junction along the road
$r_\alpha (t)$ which yields to other roads, so that
$r_\alpha (t) > \min J$. Assume that at the end of road $j$ entering
the junction (i.e.,~$j \in J$) there is no one that has the priority
over the road $r_\alpha (t)$ (i.e.,~$j < r_\alpha (t)$), i.e.,
$\bigcup_{j \in J \colon j<r_\alpha (t)} \{ \alpha' \in \{1, \ldots,
n\} \colon r_{\alpha'} (t) =j \mbox{ and } x_j (t) > L_{j}-\ell \} =
\emptyset$, and there is no one in the road where $\alpha$ is entering
(i.e.,~$\{ \alpha' \in \{1, \ldots, n\} \colon r_{\alpha'} (t) =
\mathcal{N}_\alpha\left(r_\alpha (t)\right) \} = \emptyset$). Then,
$\alpha$ drives at full speed $V_{r_\alpha (t)}$:
\begin{equation}
  \label{eq:6}
  \begin{array}{@{}r@{}c@{}l@{}}
    \dot x_\alpha
    & =
    & V_{r_\alpha(t)}
      \quad \mbox{if}
      \left\{
      \begin{array}{@{}l}
        \alpha\mbox{ is at the end of the road;}
        \\
        \alpha\mbox{'s road does not have the priority;}
        \\
        \mbox{no one is on the same road in front of }\alpha;
        \\
        \mbox{no one is on the road where $\alpha$ is going;}
        \\
        \mbox{no one is at the end of roads having priority over }\alpha.
      \end{array}
    \right.
    \\[30pt]
    & =
    & V_{r_\alpha(t)}
      \quad \mbox{if}
      \left\{
      \begin{array}{@{}l}
        x_\alpha (t) > L_{r_\alpha (t)}  - \ell;
        \\
        r_\alpha (t) > \min J;
        \\
        \left\{
        \alpha' \in \{1, \ldots, n\} \colon
        r_{\alpha'} (t) = r_\alpha (t) \mbox{ and } x_{\alpha'} (t) > x_\alpha (t)
        \right\} = \emptyset;
        \\
        \left\{
        \alpha' \in \{1, \ldots, n\} \colon
        r_{\alpha'} (t) = \mathcal{N}_\alpha\left(r_\alpha (t)\right)
        \right\} = \emptyset;
        \\
        \displaystyle
        \bigcup_{j \in J \colon j<r_\alpha (t)} \{ \alpha' {\in} \{1,
        \ldots, n\} \colon r_{\alpha'} (t) {=} j
        \mbox{ and } x_j (t) {>} L_{j}-\ell \} {=}
        \emptyset.
      \end{array}
    \right.
  \end{array}
\end{equation}

As soon as another driver, say $\alpha'$, is present near to the end
of road $j' = r_{\alpha'} (t)$
(i.e.,~$x_{\alpha'} (t) \in [L_{j'}-\ell, L_{j'}]$) entering the
junction (i.e.,~$j' \in J$) and having priority over the
$r_\alpha (t)$ road (i.e.,~$j' = r_{\alpha'} (t) < r_\alpha (t)$), the
$\alpha$th driver has to yield to $\alpha'$ and stop, see~\eqref{eq:6}
and Figure~\ref{fig:3t01} (left).
\begin{figure}[h!]
  \centering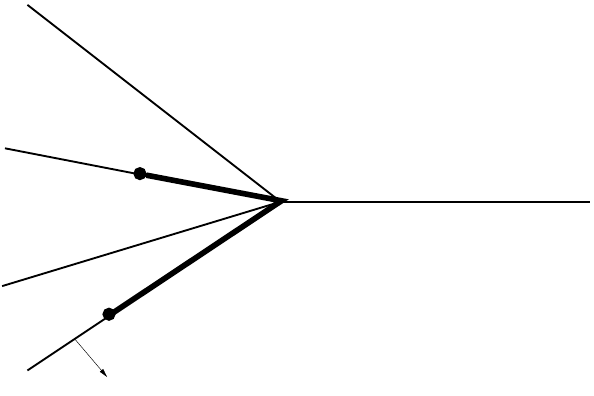\qquad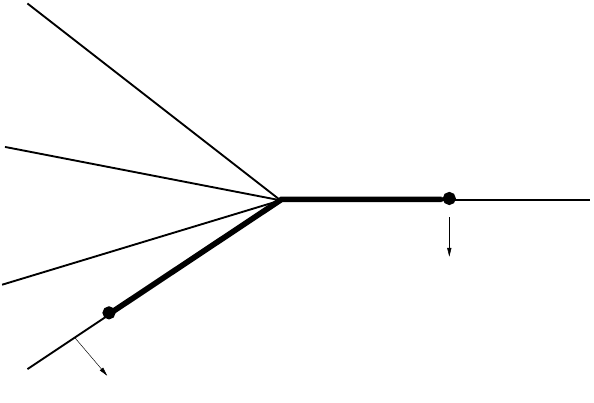
  \caption{\label{fig:3t01}Left: notation used in~\eqref{eq:7}: the
    $\alpha$th driver is near the end of the $r_\alpha(t)$ road and
    gives way to the $\alpha'$th driver, since
    $r_{\alpha'} (t) < r_\alpha (t)$. Right: notation used
    in~\eqref{eq:5}: the $\alpha$th driver adjusts his/her speed to
    the position $p$ of the nearest driver in front of him/her on road
    $\mathcal{N}_\alpha\left(r_\alpha (t)\right)$.}
\end{figure}
\begin{equation}
  \label{eq:7}
  \begin{array}{@{}r@{}c@{}l@{}}
    \dot x_\alpha
    & =
    & 0
      \quad\mbox{if}
      \left\{
      \begin{array}{@{\,}l}
        \alpha\mbox{ is at the end of the road;}
        \\
        \alpha\mbox{'s road does not have the priority;}
        \\
        \mbox{no one is on the same road in front of }\alpha;
        \\
        \mbox{someone is at the end of roads having priority over }\alpha.
      \end{array}
    \right.
    \\[25pt]
    & =
    & 0
      \quad\mbox{if}
      \left\{
      \begin{array}{@{\,}l}
        x_\alpha (t) > L_{r_\alpha (t)}  - \ell;
        \\
        r_\alpha (t) \neq \min J;
        \\
        \left\{
        \alpha' \in \{1, \ldots, n\} \colon
        r_{\alpha'} (t) = r_\alpha (t) \mbox{ and } x_{\alpha'} (t) > x_\alpha (t)
        \right\} = \emptyset;
        \\
        \displaystyle
        \bigcup_{j \in J \colon j<r_\alpha (t)} \{ \alpha' \in \{1,
        \ldots, n\} \colon r_{\alpha'} (t) =j \mbox{ and } x_j (t) > L_{j}-\ell \}
        \neq \emptyset.
      \end{array}
    \right.
  \end{array}
\end{equation}

Finally, consider the case when no one is present on the road having
priority over the road, indexed by $r_\alpha (t)$, where the
$\alpha$th driver is moving
(i.e.,~$\bigcup_{j \in J \colon j<r_\alpha (t)} \{ \alpha' \in \{1,
\ldots, n\} \colon r_{\alpha'} (t) =j \mbox{ and } x_j (t) >
L_{j}-\ell \} = \emptyset$), but other vehicles are present on the
$\mathcal{N}_\alpha\left(r_\alpha (t)\right)$ road where $\alpha$ is
heading
(i.e.,~$\{ \alpha' \in \{1, \ldots, n\} \colon r_{\alpha'} (t) =
\mathcal{N}_\alpha\left(r_\alpha (t)\right) \} \neq \emptyset$), see
Figure~\ref{fig:3t01} (right). Then, the $\alpha$th driver adapts
his/her speed to the vehicle in front of him/her:
\begin{equation}
  \label{eq:5}
  \begin{array}{@{}r@{}c@{}l@{}}
    \dot x_\alpha
    & =
    & v_{r_\alpha (t)} \! \! \left(\frac{\ell}{p + L_{r_\alpha (t)}-x_\alpha}\right)
      \mbox{if}
      \left\{
      \begin{array}{@{\,}l}
        \alpha\mbox{ is at the end of the road;}
        \\
        \alpha\mbox{'s road does not have the priority;}
        \\
        \mbox{no one is on the same road in front of }\alpha;
        \\
        \mbox{no one is at the end of roads having priority over }\alpha;
        \\
        \mbox{someone is on the road where $\alpha$ is going;}
        \\
        p
        \mbox{ is the position of the nearest vehicle in front of}
        \\
        \alpha \mbox{  on the same road.}     \end{array}
    \right.
    \\[50pt]
    & =
    & v_{r_\alpha (t)} \! \! \left(\frac{\ell}{p + L_{r_\alpha (t)}-x_\alpha}\right)
      \mbox{if}
      \left\{
      \begin{array}{@{\,}l}
        x_\alpha (t) > L_{r_\alpha (t)}  - \ell;
        \\
        r_\alpha (t) > \min J;
        \\
        \left\{
        \alpha' {\in} \{1, \ldots, n\} \colon
        r_{\alpha'} (t) {=} r_\alpha (t) \mbox{ and }
        x_{\alpha'} (t) {>} x_\alpha (t)
        \right\} {=} \emptyset;
        \\
        \displaystyle
        \!\!\! \bigcup_{{j \in J} \atop {j<r_\alpha (t)}}\!\!
        \{ \alpha' \in \{1,
        \ldots, n\} \colon r_{\alpha'} (t) {=} j ,\; x_j (t) > L_{j}-\ell \}
        {=} \emptyset;
        \\
        \left\{
        \alpha' {\in} \{1, \ldots, n\} \colon
        r_{\alpha'} (t) {=} \mathcal{N}_\alpha\left(r_\alpha (t)\right)
        \right\} {\neq} \emptyset;
        \\
        p = \min\left\{
        x_{\alpha'} \in [0, L_{r_\alpha (t)}] \colon
        r_{\alpha'} (t) = \mathcal{N}_\alpha\left(r_\alpha (t)\right)
        \right\}.
      \end{array}
    \right.
  \end{array}
\end{equation}

\subsubsection*{Existence and Uniqueness of Solutions} Summarizing,
the above formulas~\eqref{eq:3}--\eqref{eq:5} define a system of $n$
ordinary differential equations, which we write
\begin{equation}
  \label{eq:8}
  \dot x_\alpha = \mathcal{V}_\alpha (t,x)
\end{equation}
for short. The definitions above ensure that
$\mathcal{V}_\alpha (t,x) \in [0,V_{\max}]$ for all $i=1, \ldots, n$,
$t \in [0,T]$ and $x \in \reali^n$.

We now introduce a condition that states the absence of collisions
among drivers.  Recall that at time $t$ driver $\alpha$ is located at
$x_\alpha(t)$ on road $r_\alpha(t)$.
\begin{description}
\item[(NoCollision)] For all $\alpha', \alpha'' \in \{1, \ldots, n\}$,
  if $r_{\alpha'} = r_{\alpha''}$, then
  $\modulo{x_{\alpha'} - x_{\alpha''}} \geq \ell$.
\end{description}

\noindent Observe that the above condition does \emph{not} rule out
the following situation. Driver $\alpha$ is located at $x_\alpha$ on
road $j=r_\alpha(t)$ and, say, very near the junction located at the
end of road $j$, so that $x_\alpha\in \left]L_j - \ell,
  L_j\right[$. Driver $\alpha'$ moves along road $j'=r_{\alpha'}(t)$,
also entering the same junction and has the priority over road $j$, so
that $j' < j$. When $\alpha'$ passes the junction, $\alpha$ is stopped
and there may well be a time at which the distance between $\alpha$
and $\alpha'$ is smaller than $\ell$, but with $\alpha$ and $\alpha'$
being on different roads, so that no actual collision takes place.

\begin{theorem}
  \label{thm:1}
  Consider a network of $m$ interconnected roads containing at least
  one Entry Road and one Exit Road. For $j=1, \ldots, m$, on road $j$
  a speed law $v_j$ satisfying~\textbf{(SpeeedLaw)} is given.  Assign
  to $n$ drivers routes $\mathcal{R}_1, \ldots, \mathcal{R}_n$
  satisfying the~\textbf{(NoLoop)} and~\textbf{(NoDeadEnd)}
  conditions. Each driver $\alpha$ is assigned an initial position
  $x_\alpha^o$ in the first road of $\alpha$'s route
  $\mathcal{R}_\alpha$ and these initial positions satisfy
  condition~\textbf{(NoCollision)}.

  \noindent Then, the system of differential equations~\eqref{eq:8}
  admits a unique solution on the time interval
  $\left[0,+\infty\right[$. Moreover, at any positive $t$, the
  positions $x_\alpha (t)$ of the drivers at time $t$ along roads
  $r_\alpha (t)$, keep satisfying condition~\textbf{(NoCollision)}.
\end{theorem}

\begin{remark}
  \label{rem:no}
  System~\eqref{eq:8} may not have good stability properties
  concerning the dependence of solutions on the initial data, which is
  consistent with the common driving experience.

  Indeed, consider the case in Figure~\ref{fig:controEx}. The Entry
  Roads~$1$ and~$2$ end in the same junction, where the Exit Road $3$
  begins. Road~$1$ yields priority to road~$2$. For simplicity, choose
  the same speed law, say $v (\rho) = 1-\rho$, along all roads.

  Fix a sufficiently small $\epsilon > 0$. At time $t = 0$, driver $1$
  is at $x^o_1 = -\ell - \epsilon$, while driver $2$ is at
  $x^o_2 = - \epsilon^2$, see Figure~\ref{fig:controEx} (left). Then,
  the solution to~\eqref{eq:8} consists in driver $1$ passing through
  the junction and with driver $2$ following.
  \begin{figure}[h!]
    \centering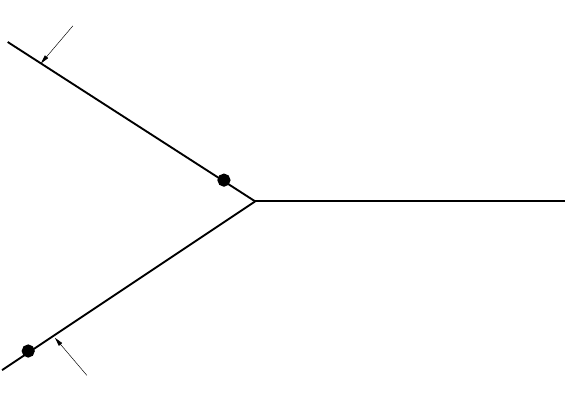\qquad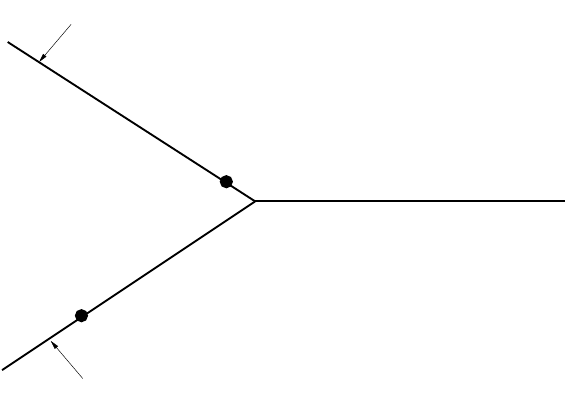
    \caption{These initial data are arbitrarily near but lead to
      solutions to~\eqref{eq:8} that are uniformly distant, due to the
      presence of priority rules.}
    \label{fig:controEx}
  \end{figure}
  On the other hand, if driver~$1$ starts from
  $x^o_1 = -\ell + \epsilon$ with driver~$2$ always at
  $x^o_2 = - \epsilon^2$, see Figure~\ref{fig:controEx} (right), then
  driver~$2$ stops owing priority to driver~$1$. The two resulting
  solutions are uniformly different as $\epsilon \to 0$.
\end{remark}

\section{Emergence of Braess Paradox}
\label{sec:BP}

In this section we show the emergence of Braess paradox in a
non-stationary setting, obtained within the framework of the system of
differential equations~\eqref{eq:8} on the network depicted in
Figure~\ref{fig:Braess}.

The seven roads are numbered as in Figure~\ref{fig:Braess} and the
Middle Roads are assigned the lengths
$L_2 = L_3 = L_5 = L_6 = \sqrt{2}$ and $L_4 = 2$. We consider the
routes
\begin{equation}
  \label{eq:13}
  \mathcal{R}_0 = [1, \, 3, \, 6, \, 7],
  \qquad
  \mathcal{R}_1 = [1, \, 2, \, 5, \, 7],
  \quad\mbox{ and } \quad
  \mathcal{R}_2 = [1, \, 3, \, 4, \, 5,\, 7] ,
\end{equation}
using the following priorities:
\begin{equation}
  \label{eq:11}
  \begin{array}{@{\,}l}
    \mbox{road }4 \mbox{ has the priority over road }2,
    \\
    \mbox{road }5 \mbox{ has the priority over road }6.
  \end{array}
\end{equation}
This means that the route $\mathcal{R}_2$ with the Middle Road $4$,
has priority over the other routes. Along road $j$ we use the speed
law $v_j$, for $j=1, \ldots, 7$, where
\begin{equation}
  \label{eq:12}
  \begin{array}{@{\,}@{}r@{\,}c@{\,}l@{\ }r@{\,}c@{\,}l@{\ }r@{\,}c@{\,}l}
    v_1 (\rho)
    & =
    & 0.9 \; (1-\rho),
    & v_2 (\rho)
    & =
    & 0.6 \; \sqrt{1-\rho},
    & v_3 (\rho)
    & =
    & (1-\rho)^{10},
    \\
    v_4 (\rho)
    & =
    & 8.0 \; (1-\rho),
    &
      v_5 (\rho)
    & =
    & 1.2 \; (1-\rho)^6,
    & v_6 (\rho)
    & =
    & \sqrt{1-\rho},
  \end{array}
  \qquad
  v_7 (\rho) = 1-\rho\,.
\end{equation}
The vehicles' length is $\ell = 0.1$.  We consider $n = 180$ drivers
leaving at time $t = 0$ from positions $x^o_1, \ldots, x^o_{180}$
evenly spaced in the interval $[-36, -0.1]$. Through a random number
generator, we randomly assign the route to each driver according to
the proportions $\theta_0, \theta_1, \theta_2$, $\theta_k$ being the
percentage of driver following the route $\mathcal{R}_k$.  Thus
$\theta_k\in[0,1]$ with $\sum_k \theta_k = 1$.

By means of Euler polygonals, with time step $h = 0.01$, we compute
(approximate) solutions to~\eqref{eq:8}. Each integration is repeated
$20$ times with different route assignments to the drivers, but
assigning the same frequencies $\theta_0, \theta_1$, and
$\theta_2$. For each driver $\alpha$, we compute the travel time as
the first time step when $\alpha$ is on road $7$. Then, all travel
times are averaged over the drivers following the same route, and the
results are displayed in Table~\ref{tab:1}.

\begin{table}[h!]
  \begin{center}
    \noindent\!\!%
    \begin{tabular}{@{}rrr|rrr|rrr|@{}r@{}}
      \multicolumn{3}{c|}{Assigned Distrib.}
      & \multicolumn{3}{c|}{Effective Distribution}
      & \multicolumn{3}{c|}{Travel Time}
      & \multicolumn{1}{c}{Mean}
      \\
      $\theta_0$
      & $\theta_1$
      & $\theta_2$
      & $\Theta_0$
      & $\Theta_1$
      & $\Theta_2$
      & $T_0$
      & $T_1$
      & $T_2$
      & $\sum_i \! \Theta_i  T_i$
      \\[2pt] \hline
      0.00
      & 0.00
      & 1.00
      &  0.0000
      & 0.0000
      & 1.0000
      & //
      & //
      & 105.4
      & 105.4
      \\ \hline
      0.05
      & 0.05
      & 0.90
      & 0.05222
      & 0.05028
      & 0.8975
      & \textbf{105.6}
      & \textbf{106.7}
      & \textbf{100.1}
      & 100.7
      \\
      0.06
      & 0.06
      & 0.88
      & 0.05833
      & 0.05861
      & 0.8831
      & \textbf{102.7}
      & \textbf{100.3}
      & \textbf{99.33}
      & 99.58
      \\
      0.07
      & 0.07
      & 0.86
      & 0.07056
      & 0.07028
      & 0.8592
      & \textbf{100.0}
      & \textbf{101.6}
      & \textbf{98.21}
      & 98.58
      \\
      0.06
      & 0.04
      & 0.90
      & 0.05917
      & 0.03444
      & 0.9064
      & 101.7
      & 106.4
      & 101.9
      & 102.1
      \\
      0.04
      & 0.06
      & 0.90
      & 0.04333
      & 0.06167
      & 0.8950
      & 100.0
      & 95.95
      & 99.10
      & 98.95
      \\
      0.30
      & 0.30
      & 0.40
      & 0.3083
      & 0.2761
      & 0.4156
      & 76.99
      & 77.55
      & 79.84
      & 78.33
      \\ \hline
      0.45
      & 0.45
      & 0.10
      & 0.4467
      & 0.4486
      & 0.1047
      & \textbf{63.06}
      & \textbf{65.45}
      & \textbf{60.79}
      & 63.89
      \\
      0.47
      & 0.47
      & 0.06
      & 0.4761
      & 0.4633
      & 0.06056
      & \textbf{61.93}
      & \textbf{62.63}
      & \textbf{60.02}
      & 62.14
      \\
      0.50
      & 0.50
      & 0.00
      & 0.4983
      & 0.5017
      & 0.0000
      & 58.45
      & 60.00
      & //
      & 59.23
    \end{tabular}
  \end{center}
  \caption{Sample results obtained from integrating~\eqref{eq:8}. The
    travel times $T_k$, for $k=0, 1, 2$, are the averages of the times
    at which drivers following route $k$ enter road~$7$. The mean
    travel time is $\sum_{k=0}^2 \Theta_k T_k$, where $\Theta_k$ is
    the actual portion of drivers following route $k$.\label{tab:1}}
\end{table}

Here, the \emph{travel time} $T_k$ is the average time that drivers
following route $k$ need to reach road~$7$.

The bold travel times in Table~\ref{tab:1} display situations fully
coherent with Braess paradox and with
$(\theta_0, \theta_1, \theta_2) = (0, \, 0, \, 1)$ being a Nash
equilibrium for the travel times.  Note also that all displayed
integrations are consistent with a \emph{weak}, but still surprising,
form of the Braess paradox, in the sense that the overall mean travel
times with the new road $4$ being present are all clearly larger than
the mean travel time without road $4$.

Figure~\ref{fig:queue} displays a sample integration of the model
described by~\eqref{eq:8} with speed laws~\eqref{eq:12}, where we can
see the effect of the priority of route
$\mathcal{R}_2 = [1, \, 3, \, 4, \, 5,\, 7]$, which contain the new
road $4$, over the other routes. As a consequence, a queue is formed
in road $6$.

\begin{figure}[!h]
  \includegraphics[width=0.5\textwidth,trim = 50 25 5
  70]{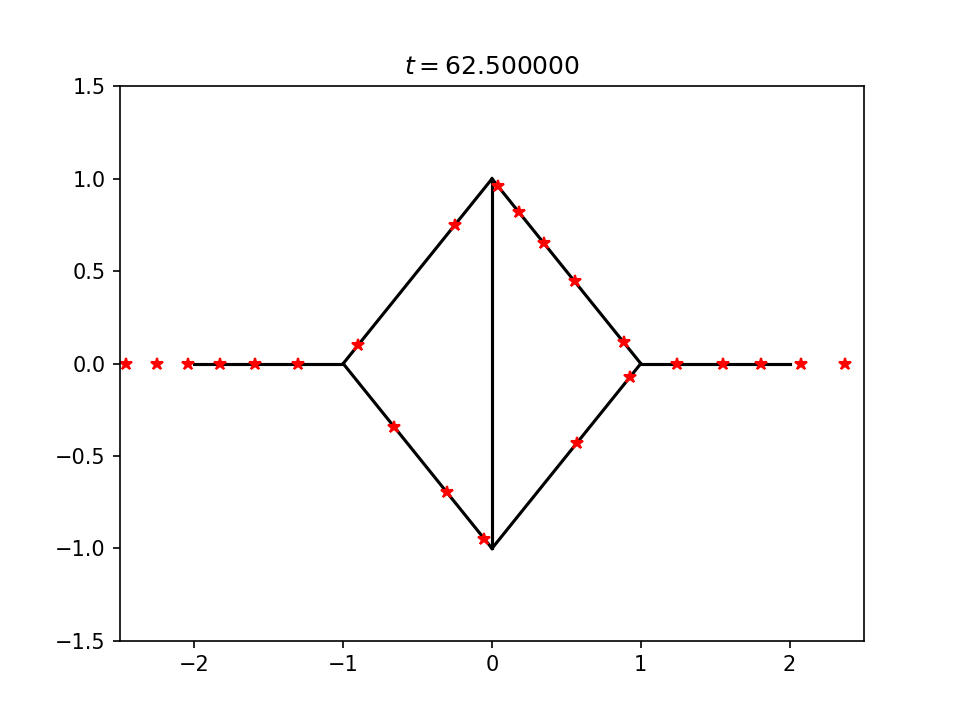}\ \ \
  \includegraphics[width=0.5\textwidth,trim = 50 25 5 70]{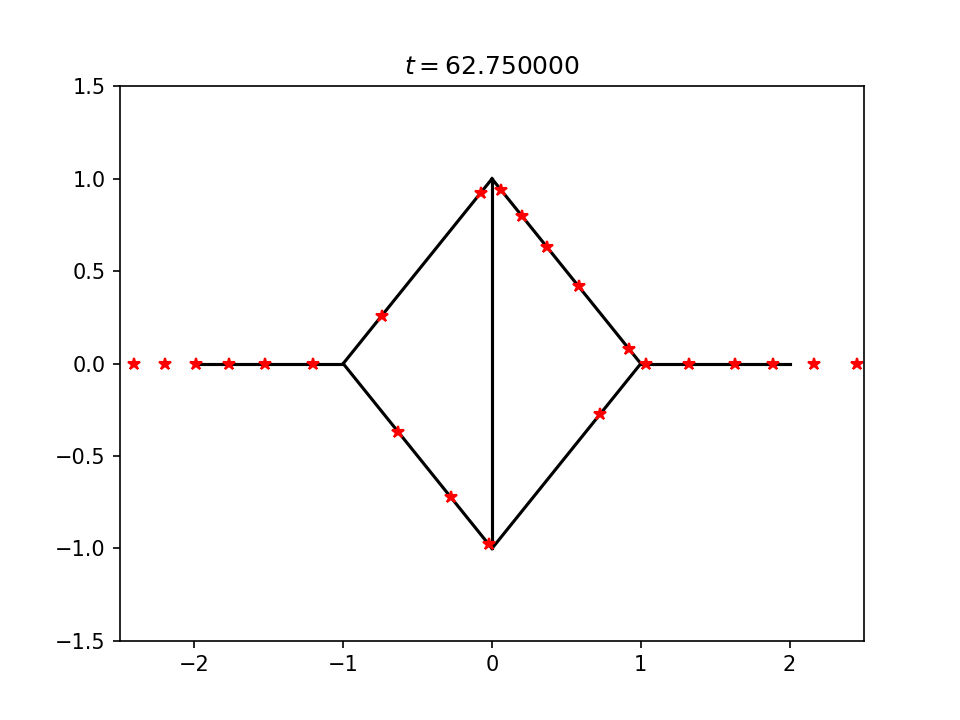}\\[30pt]
  \includegraphics[width=0.5\textwidth,trim = 50 25 5
  70]{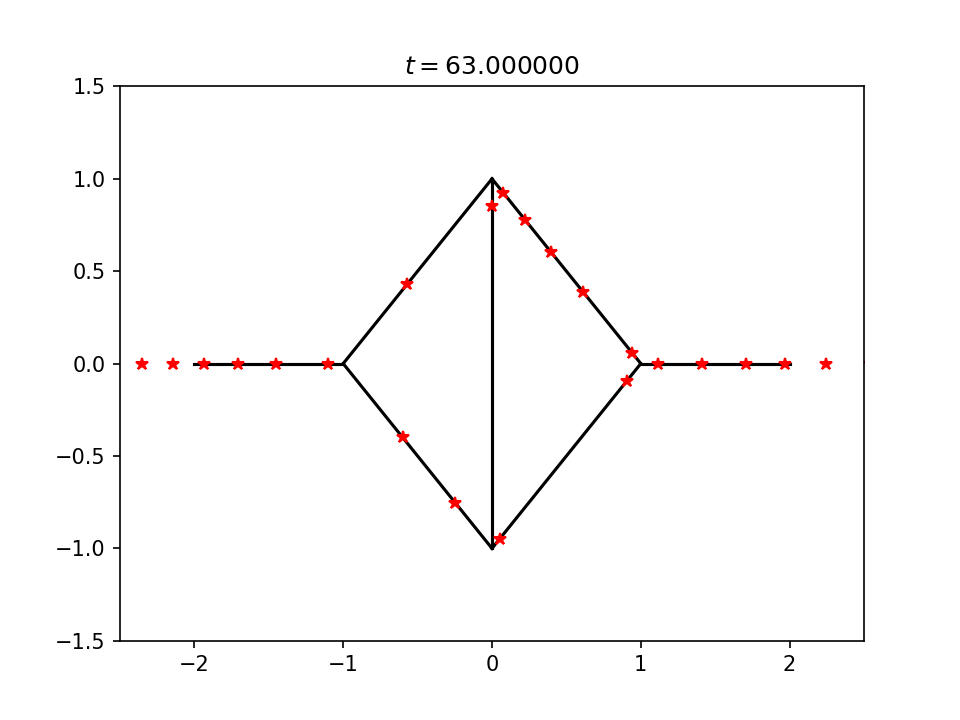}\ \ \
  \includegraphics[width=0.5\textwidth,trim = 50 25 5 70]{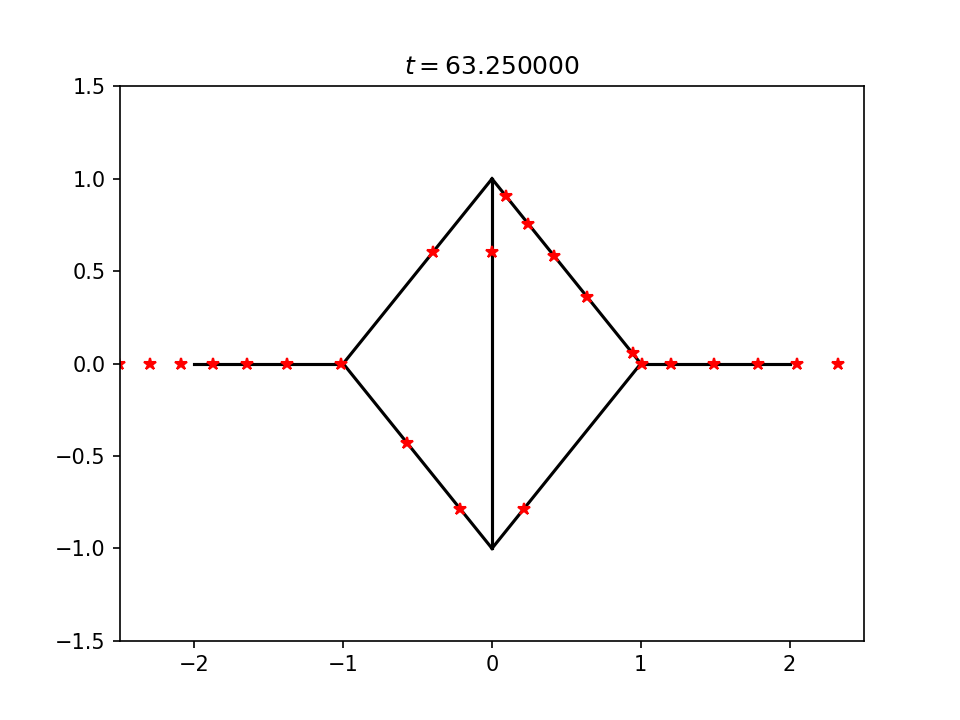}\\[30pt]
  \includegraphics[width=0.5\textwidth,trim = 50 25 5
  70]{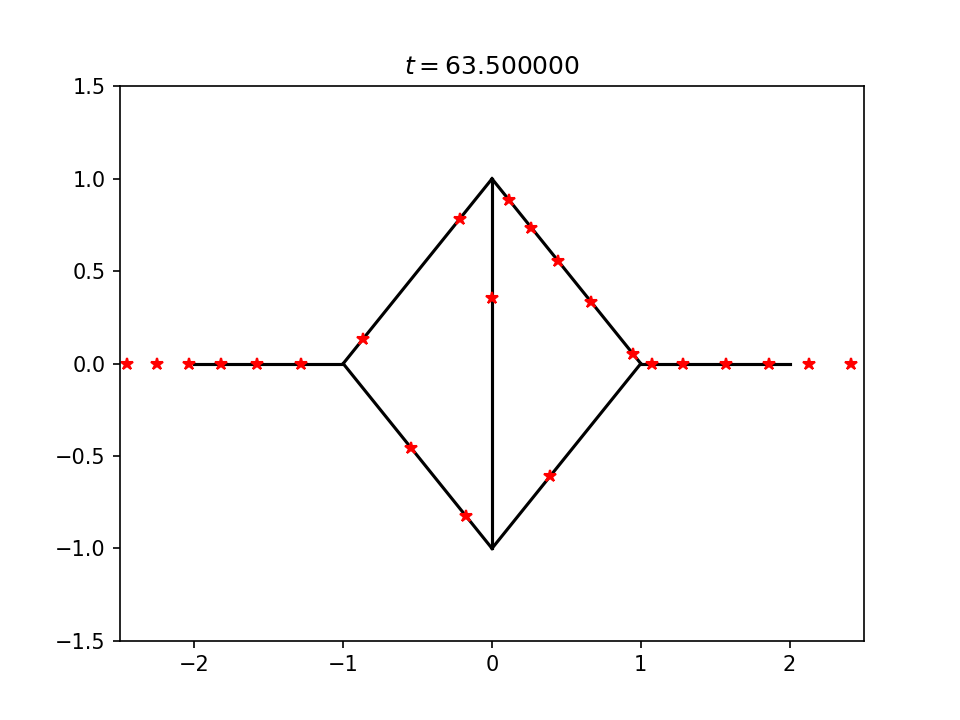}\ \ \
  \includegraphics[width=0.5\textwidth,trim = 50 25 5 70]{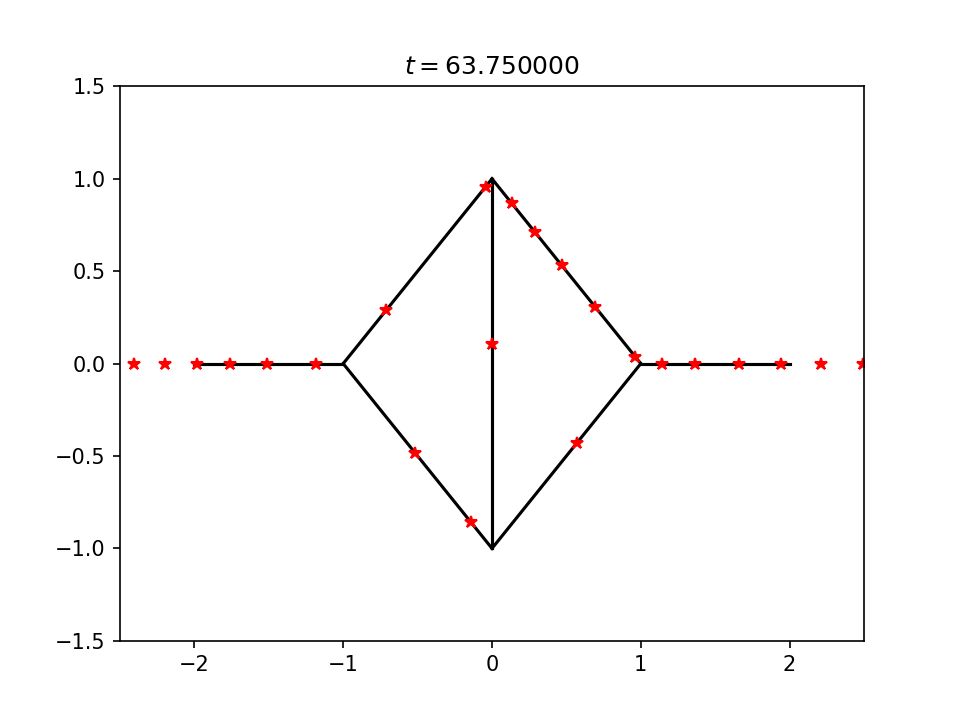}\\
  \caption{Sample integration of the model described
    by~\eqref{eq:8}. Note that the top right road yields priority to
    the bottom right one. As a consequence, all during the (short)
    time interval displayed, the 5 vehicles in the top right road line
    up and give priority to the vehicles coming from their
    right.\label{fig:queue}}
\end{figure}

\section{Analytic Proofs}

The following lemma tackles the basic local existence part of
Theorem~\ref{thm:1}.

\begin{lemma}
  \label{lem:1}
  With the assumptions and notations of Theorem~\ref{thm:1}, assume
  that at time $\bar t \geq 0$ the drivers are distributed along the
  network at positions $\bar x_\alpha (\bar t)$ satisfying
  condition~\textbf{(NoCollision)}. Then, there exists a positive
  $\epsilon$ and uniquely determined functions
  $x_\alpha \colon [\bar t, \bar t + \epsilon] \to \reali$
  solving~\eqref{eq:8}. Moreover, for all
  $t \in [\bar t, \bar t + \epsilon]$,
  condition~\textbf{(NoCollision)} holds.
\end{lemma}

\begin{proof}
  The proof is divided in three steps.

  \bigskip

  \noindent\textbf{1.} For all $\alpha$ such that $r_\alpha (\bar t)$
  is an Exit Road, the function $x_\alpha$ can be uniquely defined on
  $\left[\bar t, + \infty\right[$ solving a standard Follow-the-Leader
  ODE system. Note that, by the standard properties of this model,
  these $x_\alpha$ satisfy the~\textbf{(NoCollision)} condition.

  \bigskip

  \noindent\textbf{2.}  If Middle Roads and Entry Roads are empty, the proof
  is finished.

  Otherwise, introduce the set of Entry and Middle Roads where there
  is at least one driver at time $\bar t$ whose next road is an Exit
  Road:
  \begin{displaymath}
    J_1 = \{j \colon \exists \alpha \in \{1, \ldots, n\}
    \mbox{ with } r_\alpha (\bar t) = j
    \mbox{ and }
    \mathcal{N}_\alpha\left(r_\alpha (\bar t)\right)
    \mbox{ is an Exit Road}
    \} \,.
  \end{displaymath}
  Each road $j \in J_1$ ends at a junction where an Exit Road
  begins. Consider one of these junctions, say $C$ and call
  $j_1, \ldots, j_k$ the roads in $J_1$ entering $C$. We may assume
  that $j_1 < j_2 < \cdots < j_k$, so that road $j_1$ has the
  priority. The drivers along road $j_1$ are at positions
  $x_{\alpha_1} (\bar t), x_{\alpha_2} (\bar t), \cdots,
  x_{\alpha_\nu} (\bar t)$, with
  $L_{j_1} > x_{\alpha_1} (\bar t)> x_{\alpha_2} (\bar t) > \cdots >
  x_{\alpha_\nu} (\bar t) \geq 0$. The trajectory of driver $\alpha_1$
  is uniquely determined, since all trajectories along Exit Roads are
  known. Therefore, along road $j_1$, the usual Cauchy theorem for
  ODEs ensures the existence and uniqueness of a solution
  to~\eqref{eq:8} at least on the time interval
  $[\bar t, \bar t + \epsilon]$, where
  $\epsilon = (L_{j_1} - x_{\alpha_1} (\bar t)) / V_{\max}$ and, by
  construction, condition~\textbf{(NoCollision)} holds.

  Assume now that all drivers' trajectories along roads
  $j_1, \ldots, j_{h-1}$ are uniquely defined on the time interval
  $[\bar t, \bar t+ \epsilon]$, for a positive $\epsilon$. Denote by
  $\alpha_1, \ldots, \alpha_\nu$ the drivers on road $j_h$, with
  $L_{j_h} > x_{\alpha_1} (\bar t)> x_{\alpha_2} (\bar t) > \cdots >
  x_{\alpha_\nu} (\bar t) \geq 0$. The speed of $\alpha_1$ is a unique
  non-negative $\L1$ function defined at least on the time interval
  $[\bar t, \bar t + \epsilon']$, where
  $\epsilon' = \min \{\epsilon, (L_{j_h} - x_{\alpha_1} (\bar
  t))/V_{\max}\}$, so that the trajectory of $\alpha_1$
  solves~\eqref{eq:8} and is uniquely defined. Iteratively, the same
  holds first for the trajectories of $\alpha_2, \ldots, \alpha_\nu$
  and then along all other roads $j_{h+1}, \ldots, j_k$, always
  complying with condition~\textbf{(NoCollision)}.

  \bigskip

  \noindent\textbf{3.} By condition~\textbf{(NoLoop)}, the above
  procedure can be iterated, covering the whole network and without
  considering the same interval twice. Indeed, consider the set of
  roads entering $J_i$:
  \begin{displaymath}
    J_{i+1} = \{j \colon \exists \alpha \in \{1, \ldots, n\}
    \mbox{ with } r_\alpha (\bar t) = j
    \mbox{ and }
    \mathcal{N}_\alpha\left(r_\alpha (\bar t)\right) \in J_i
    \}
  \end{displaymath}
  and proceed exactly as in the step~\textbf{2} above.

  \bigskip

  Here, a unique solution to~\eqref{eq:8} was constructed on the time
  interval $[\bar t, \bar t + \epsilon_*]$, complying with
  condition~\textbf{(NoCollision)}, where $\epsilon_*$ is the minimum
  of a finite quantity of positive numbers. The proof is completed.
\end{proof}

Below, for each driver, we also use the time-dependent coordinate
$y_\alpha (t)$, which quantifies the total distance driven by the
driver $\alpha$ at time $t$. For instance, with reference to
Figure~\ref{fig:Braess}, if the driver $\alpha = 2$ follows the route
$\mathcal{R}_2 = [0, 2, 3, 4, 6]$, starting from
$x^o_2 \in \left]-\infty,0\right[$ in the Entry Road $j = 0$ at time
$0$ and at time $t$ is moving along road $4$, then $r_2 (t) = 4$ and
$y_2 (t) = \modulo{x^o_2} + L_2 + L_3 + x_2 (t)$, with
$x_2 (t) \in \left[0,L_4\right[$.

Given the route $\mathcal{R}_\alpha = [j_0, j_1, \ldots, j_k]$ (with
road lengths $L_{j_0}, L_{j_1}, \ldots $,) for driver $\alpha$, the
initial position $x^o_\alpha$ and $x_\alpha (t)$, the length
$y_\alpha (t)$ covered by the $\alpha$th driver is uniquely
determined. Indeed, if at time $t$ the $\alpha$ driver is along road
$j_{i_*}=r_\alpha(t)$, we have
\begin{equation}
  \label{eq:10}
  y_\alpha (t) =
  \left\{
    \begin{array}{@{\,}ll}
      \modulo{x_\alpha^o} + \sum_{i < i_*} L_{j_i} + x_\alpha (t)
      & \alpha \mbox{ starts from an Entry Road,}
      \\
      L_{i_1} - x_\alpha^o + \sum_{i < i_*} L_{j_i} + x_\alpha (t)
      & \alpha \mbox{ starts from a Middle Road,}
      \\
      x_\alpha (t) - x^o_\alpha
      & \alpha \mbox{ starts from an Exit Road.}
    \end{array}
  \right.
\end{equation}
The inverse correspondence is straightforward.

\begin{proofof}{Theorem~\ref{thm:1}}
  By Lemma~\ref{lem:1}, for given initial data, problem~\eqref{eq:8}
  admits a unique solution on the interval $[0, \epsilon_0]$, for a
  positive $\epsilon_0$.

  Prolong, in a unique way, the solution to~\eqref{eq:8} on the time
  interval $[\epsilon_0, \epsilon_0 + \epsilon_1]$ applying
  Lemma~\ref{lem:1}.

  We claim that a solution to~\eqref{eq:8} can be uniquely constructed
  on all $\left[0, +\infty\right[$. Indeed, assume (by contradiction)
  that the above procedure yields a solution $x_\alpha$, for
  $\alpha \in \{1, \ldots, n\}$ defined on the maximal time interval
  $\left[0, T\right[$, for a positive $T$. For all
  $\alpha \in \{1, \ldots, n\}$, the corresponding function
  $t \mapsto y_\alpha (t)$ is defined on $\left[0, T\right[$ and it is
  Lipschitz continuous, hence it is uniformly continuous and can be
  uniquely extended by continuity to the time interval $[0,T]$. As a
  consequence, also $x_\alpha$ can be uniquely extended to the whole
  interval $[0,T]$. At time $T$, we thus apply again
  Lemma~\ref{lem:1}, obtaining a solution defined on
  $\left[0, T+\epsilon_*\right[$, for a positive $\epsilon_*$. This
  contradicts the maximality of the above choice of $T$.
\end{proofof}

\section{Conclusions}
\label{sec:C}

This paper provides the analytic framework to use microscopic traffic
model on road networks. Traffic at junctions is ruled by fixed
priority rules, so that queues may form and disappear, depending on
the overall traffic distribution. Existence and uniqueness of
solutions is proved, while continuous dependence may fail, which is
consistent with everyday experience. Moreover, vehicles may not
collide, once the initial datum assigned is reasonable.

This framework is then used to describe a non-stationary instance of
Braess paradox. Adding a very fast road to an existing network may
increase the travel times. Here, a game theoretic approach was used,
each driver being a player aiming at minimizing his/her travel time.

On the basis of the present results, further questions arise and can
be tackled. A very appealing research direction concerns the
\emph{control} of network traffic. For instance,
following~\cite{BressanHan2011, BressanHan2012}, can the introduction
of a suitable toll avoid the insurgence of Braess paradox? Once
Theorem~\ref{thm:1} is extended to time dependent priority rules
(i.e., traffic lights), which seems a merely technical issue, is it
possible to find optimal timings at the junctions that minimize travel
times?

\section*{Acknowledgments}
RMC and FM were partially supported by the INdAM-GNAMPA 2019 project
\emph{Partial Differential Equations of Hyperbolic or Nonlocal Type
  and Applications}. The research of HH was supported by the grant
{\it Waves and Nonlinear Phenomena (WaNP)} from the Research Council
of Norway. The \emph{IBM Power Systems Academic Initiative}
substantially contributed to the numerical integrations.

\newpage
\bibliographystyle{siamplain}
\bibliography{ColomboHoldenMarcellini_revised}

\end{document}